\documentclass[a4paper,UKenglish]{lipics-v2019}

\usepackage{microtype}

\usepackage[utf8]{inputenc}
\usepackage{amsmath}
\usepackage{amsfonts}
\usepackage{amsthm}
\usepackage{bm}
\usepackage{algorithm}
\usepackage{algorithmic}
\usepackage{comment}

\renewcommand\proofname{\bf Proof}

\newtheorem{lem}{Lemma}

\renewcommand{\root}{\textrm{root}}
\newcommand{\OPT}{\textrm{OPT}}
\newcommand{\DP}{\texttt{DP}}

\title{PTAS and Exact Algorithms for $r$-Gathering Problems on Tree}

\author{Soh Kumabe}{The University of Tokyo \\ RIKEN AIP}{soh\_kumabe@mist.i.u-tokyo.ac.jp}{}{}
\author{Takanori Maehara}{RIKEN AIP}{takanori.maehara@riken.jp}{}{}

\authorrunning{S. Kumabe and T. Maehara}

\Copyright{S. Kumabe and T. Maehara}

\ccsdesc[500]{Mathematics of computing~Combinatorial optimization}

\keywords{$r$-Gathering Problem; Tree; Polynomial-Time Approximation Scheme}

\category{}

\relatedversion{}

\supplement{}

\funding{}

\acknowledgements{}

\EventEditors{John Q. Open and Joan R. Access}
\EventNoEds{2}
\EventLongTitle{42nd Conference on Very Important Topics (CVIT 2016)}
\EventShortTitle{CVIT 2016}
\EventAcronym{CVIT}
\EventYear{2016}
\EventDate{December 24--27, 2016}
\EventLocation{Little Whinging, United Kingdom}
\EventLogo{}
\SeriesVolume{42}
\ArticleNo{23}
\nolinenumbers 
\hideLIPIcs  

\begin{document}
\maketitle

\begin{abstract}
$r$-gathering problem is a variant of facility location problems. In this problem, we are given a set of users and a set of facilities on same metric space. We open some of the facilities and assign each user to an open facility, so that at least $r$ users are assigned to every open facility. We aim to minimize the maximum distance between user and assigned facility. In general, this problem is NP-hard and admit an approximation algorithm with factor $3$ \cite{armon2011min}. It is known that the problem does not admit any approximation algorithm within a factor less than $3$ \cite{armon2011min}. In our another paper, we proved that this problem is NP-hard even on spider, which is a special case of tree metric \cite{wareware}. In this paper, we concentrate on the problems on a tree. First, we give a PTAS for $r$-gathering problem on a tree. Furthermore, we give PTAS for some variants of the problems on a tree, and also give exact polynomial-time algorithms for another variants of $r$-gathering problem on a tree.

\end{abstract}

\section{Introduction}

\subsection{Background and Motivation}

In \emph{min-max $r$-gathering problem}, we are given a metric space $\mathcal{M}$ that contains several users $\mathcal{U}$ and facilities $\mathcal{F}$.
We can open some facilities and assign each user to an opened facility so that each opened facility has so least $r$ users.
The objective of the problem is to minimize the maximum distance between the facilities and the assigned users~\cite{armon2011min}.

This problem has an application in shelter evacuation problem~\cite{akagi2015r}: 
There are people and evacuation shelters, and we divide the people into shelters so that all people can evacuate in minimum possible time.
Each shelter must have at least $r$ people to maintain their lives in shelters.
The problem also has an application to privacy protection~\cite{sweeney2002k}.
A set of clusters satisfies \emph{$k$-anonymity} if each cluster has at least $k$ users; this condition prevents reconstructing personal information from the clustering. 

Several tractability and intractability results are known.
There is a polynomial-time $3$-approximation algorithm for a general metric space $\mathcal{M}$,
and no better approximation ratio can be achieved unless P=NP~\cite{armon2011min}.
If $\mathcal{M}$ is a line, we can solve the problem exactly by dynamic programming (DP)~\cite{akagi2015r,han2016r,nakano2018simple}, where the fastest algorithm runs in linear-time~\cite{sarker2019r}.
When $\mathcal{M}$ is a \emph{spider}, which is a metric space constructed by joining half-lines at their endpoints, Ahmed et al.~\cite{ahmed2019r} proposed a fixed-parameter tractable algorithm parameterized by $r$ and the degree of the center.
In our co-submitted paper~\cite{wareware}, the authors showed the problem is NP-hard if $\mathcal{M}$ is a spider, and the problem admits a fixed-parameter tractable algorithm parameterized by $r$.

\subsection{Our Contribution}

The goal of this study is to explore the boundary of tractability of the min-max $r$-gathering problem.
Specifically, we consider the problem on \emph{tree}, which is a natural graph class that contains spiders as a subclass.

It is easy to see that the problem does not admit a fully polynomial-time approximation scheme (FPTAS) (see Proposition~\ref{nofptas}).
Therefore, the best-possible positive result that we can expect is a polynomial-time approximation scheme (PTAS). 
Our main contribution is to establish PTAS for this problem as follows.
\begin{theorem}
\label{ptas}
There exists an algorithm for the min-max $r$-gathering problem on a tree so that for any $\epsilon>0$ it outputs a solution with an approximation ratio of $1+\epsilon$ in $(|\mathcal{U}|+|\mathcal{F}|)^{O(1/\epsilon)}$ time.
\end{theorem}
The proposed algorithm seeks the optimal value by a binary search, and in each step, it solves the corresponding decision problem by a DP on a tree.
Here, the most difficult part is establishing an algorithm for the decision problem.

This technique can also be applied to other problems, for example, $(r,\epsilon)$-gathering problem and $r$-gathering problem with a constraint on the number of open facilities. It can also be shown that these problems are NP-hard and do not admit FPTAS unless P=NP by the same reduction. Thus, these are also tight results.

On the other hand, there are variants of $r$-gathering, which can be solved exactly in polynomial time on a tree. We provide polynomial time algorithms via DP for two problems: min-sum $r$-gathering problem and min-max (and min-sum) $r$-gathering with proximity requirement.

\subsection{Organization}

The rest of the paper is organized as follows. In section 2, we give a PTAS for min-max $r$-gathering problem on a tree. We also show the problems which admit essentially same PTAS. In subsection 3.1, we provide the polynomial-time algorithm which solves the min-sum version of $r$-gathering problem exactly on a tree. Finally, in 3.2, we provide the polynomial-time algorithm which solves the min-max (and min-sum) $r$-gathering with proximity requirement exactly on a tree.

\section{PTAS for min-max $r$-Gathering on Tree}

A \emph{weighted tree} $T = (V(T), E(T); l)$ is an undirected connected graph without cycles, where $V(T)$ is the set of vertices, $E(T)$ is the set of edges, and $l \colon E(T) \to \mathbb{R}_+$ is the non-negative edge length. 
$T$ forms a metric space by the tree metric $d(v, w)$, which is the sum of the edge lengths on the unique simple $v$-$w$ path for any vertices $v, w \in V(T)$.
We consider the min-max $r$-clustering problem on this metric space.

Without loss of generality, we assume that all users and facilities are located on different vertices; otherwise, we add new vertices connected with edges of length zero and separate the users/facilities into the new vertices. 
By performing similar operations, we also assume that $T$ is a rooted full binary tree rooted at a special vertex $root$ (that is, we can make $T$ to the rooted tree so that every vertex has zero or two children).
These operations only increase the number of vertices (and edges) of tree by a constant factor; these do not affect the time complexity of our algorithms.
We denote the subtree of $T$ rooted at $v$ by $T_v$.

\subsection{Hardness of the Problem}

We first see that the problem does not admit FPTAS. 
This is a simple consequence of our co-submitted paper~\cite{wareware} that proves the NP-hardness of the problem on a spider.
\begin{proposition}
\label{nofptas}
There is no FPTAS for the min-max $r$-gathering problem on a spider unless P=NP.
\end{proposition}
\begin{proof}
In \cite{wareware}, the authors proved that the min-max $r$-gathering problem is NP-hard even if the input is a spider and the edge lengths are integral, and the diameter of the spider is bounded by $O(n + m)$. 
Let us take such an instance.
If there is a FPTAS for the min-max $r$-gathering problem on a spider, 
by taking $\epsilon = 1 / (c(n + m))$ for sufficiently large constant $c$, we get an optimal solution because the optimal value is an integer at most $O(n+m)$.
This contradicts to the hardness.
\end{proof}

\subsection{Algorithm. Part 1: Binary Search}

In the following sections, we develop a PTAS for the problem.
We employ a standard practice for min-max problems: we guess the optimal value by binary search and solve the corresponding decision problem for the feasibility of the problem whose objective value is at most the guessed optimal value.


First, we run Armon et al.'s $3$-approximation algorithm~\cite{armon2011min} to obtain $B$ such that $B/3 \leq \OPT(\mathcal{I})\leq B$ holds.
Then we set $[B/3, B]$ as the range for the binary search.
This part is needed to run the algorithm in strongly polynomial-time.

For the binary search, we design the following oracle $\texttt{Solve}(\mathcal{I},b,\delta)$: Given an instance $\mathcal{I}$, threshold $b$, and positive number $\delta$, it reports YES if $\OPT(\mathcal{I}) \le (1 + \delta) b$, and NO if $\OPT(\mathcal{I}) > b$.
If $b < \OPT(\mathcal{I}) \le (1 + \delta) b$ then both answer is acceptable.
Our oracle also outputs the corresponding solution as a certificate if it returns YES.
Note that we cannot set $\delta = 0$ since it reduces to the decision version of the min-max $r$-gathering problem, which is NP-hard on a tree~\cite{wareware}.


\begin{algorithm}
\caption{PTAS for the Min-Max $r$-Gathering Problem on Tree}    
\label{alg}                       
\begin{algorithmic}
\REQUIRE An instance of the min-max $r$-gathering on a tree $\mathcal{I}$, positive number $\epsilon$
\STATE Run $3$-approximation algorithm for $\mathcal{I}$ and let the optimal value be $B$
\STATE $b_1=\frac{B}{3}, b_2=B$
\WHILE{$b_2-b_1>\frac{\epsilon}{9}B$}
    \STATE $b=\frac{b_1+b_2}{2}$
    \IF{$\texttt{Solve}(\mathcal{I},b,\frac{\epsilon}{2})$ returns YES}
        \STATE $b_2=b$
    \ELSE
        \STATE $b_1=b$
    \ENDIF
\ENDWHILE
\ENSURE the solution of $\texttt{Solve}(\mathcal{I},b_2,\frac{\epsilon}{2})$
\end{algorithmic}
\end{algorithm}

If we have such oracle, we can construct a PTAS as shown in Algorithm~\ref{alg}. 
The correctness of this algorithm is as follows.

\begin{lemma}
Assume that there is a deterministic strongly polynomial time oracle $\texttt{Solve}$ described above. Then, Algorithm~\ref{alg} gives a solution to the min-max $r$-gathering problem whose cost is at most $(1+\epsilon)\OPT(\mathcal{I})$ in strongly polynomial time.
\end{lemma}

\begin{proof}
By the definition of $\texttt{Solve}$ and the algorithm, during the algorithm, $\texttt{Solve}(\mathcal{I},b_2,\frac{\epsilon}{2})$ always returns YES, and $\texttt{Solve}(\mathcal{I},b_1,\frac{\epsilon}{2})$ returns NO unless $\texttt{Solve}(\mathcal{I},\frac{B}{3},\frac{\epsilon}{2})$ is YES and $b_1=\frac{B}{3}$.
Thus, we have $b_1\leq \OPT(\mathcal{I})$.
Therefore, the algorithm outputs the solution with cost at most $b_2 (1 + \frac{\epsilon}{2})$, which is at most $\OPT(\mathcal{I})(1 + \epsilon)$, because
\[
b_2 (1 + \frac{\epsilon}{2}) \leq (b_1+\frac{\epsilon}{9}B)(1+\frac{\epsilon}{2})\leq \OPT(\mathcal{I})(1+\frac{\epsilon}{3})(1+\frac{\epsilon}{2})\leq \OPT(\mathcal{I})(1+\epsilon).
\]
The algorithm terminates in $O(\log \frac{1}{\epsilon})$ steps because the gap $b_2-b_1$ becomes half in each step,
That completes proof.
\end{proof}

\subsection{Algorithm. Part 2: Rounding Distance}

In this and next subsections, we propose a DP algorithm for $\texttt{Solve}(\mathcal{I}, b, \delta)$.
Our algorithm maintains ``distance information'' in the indices of the DP table.
For this purpose, we round the distances so that all the vertices (thus the users and facilities) are located on the points which are distant from the root by distance multiple of positive number $t$ as follows.

For each edge $e = (v, w) \in E(T)$, where $v$ is closer to the root, we define the rounded length by $l'(e) = \lfloor \frac{d(\root,w)}{t}\rfloor-\lfloor \frac{d(\root,v)}{t}\rfloor$.
Intuitively, this moves all the vertices ``toward the root'' and regularize the edge lengths into integers. 
Then, we define the rounded distance $d'$ the metric on $\mathcal{I}'$.

This rounding process only changes the optimal value a little.
\begin{lemma}
For any pair of vertices $v, w$, $d(v,w)-2t\leq d'(v,w)t\leq d(v,w)+2t$ holds.
Especially, $|\OPT(\mathcal{I})-\OPT(\mathcal{I'})t|\leq 2t$.
\end{lemma}
\begin{proof}
Let $x$ be the lowest common ancestor of $v$ and $w$.
Then, $x$ is on the $v$-$w$ path; thus, $d(v,w) = d(x,v) + d(x,w)$ and $d'(v,w)=d'(x,v)+d'(x,w)$ hold.
Since $d(x,v)=d(\root,v)-d(\root,x), d'(x,v)=d'(\root,v)-d'(\root,x)$ and $d(\root,z)-t\leq d'(\root,z)t\leq d(\root,z)$ for all vertex $z$, we have $d(x,v)-t\leq d'(x,v)t\leq d(x,v)+t$. 
We also have $d(x,w)-t\leq d'(x,w)t\leq d(x,w)+t$ by symmetry.
Thus $d(v,w)-2t\leq d'(v,w)t\leq d(v,w)+2t$ holds.
Since the cost of the min-max $r$-gathering problem is the maximum length of some paths, the second statement follows from the first statement.
\end{proof}
This lemma implies that an algorithm that determines whether $\mathcal{I}'$ has a solution with cost at most $\frac{b+2t}{t}$ works as an oracle $\texttt{Solve}(\mathcal{I}, b, \epsilon)$ if $t = \frac{b \delta}{4}$.


\subsection{Dynamic Programming}

Now we propose an algorithm to determine whether $\mathcal{I}'$ has a solution with cost at most $\frac{b+2t}{t}$.
Since all the edge costs of $\mathcal{I}'$ are integral, without loss of generality, we replace the threshold by $K := \lfloor \frac{b+2t}{t}\rfloor$.
An important observation is that $K$ is bounded by a constant since $K\leq \frac{b+2t}{t}=\frac{4}{\delta}+2$.

Our algorithm is a dynamic programming on a tree.
For vertex $v$, arrays $P=(p_0, \dots, p_K)$ and $Q=(q_0, \dots, q_K)$, we define a boolean value $\DP[v][P][Q]$.
$\DP[v][P][Q]$ is $\texttt{true}$ if there is a way to 
\begin{itemize}
    \item open some facilities in $T_v$, and
    \item assign some users in $T_v$ to the opened facilities so that
    \item for all $i=0, \dots, K$ there are $p_i$ unassigned users in $T_v$ who are distant from $v$ by distance $i$ and no other users are unassigned, and
    \item for all $i=0, \dots, K$ we will assign $q_i$ users out of $T_v$ who are distant from $v$ by distance $i$ to open facilities in $T_v$,
\end{itemize}
and \texttt{false} otherwise.
$\DP[\root][(0, \dots, 0)][(0,\dots, 0)]$ is the solution to the problem.
The elements of $P$ and $Q$ are non-negative integers at most $n$; thus, the number of the DP states is $|V(T)| \times (n+1)^{2(K+1)}$, which remains in polynomial in the size of input.

The remaining task is to write down the transitions.
For arrays $X$ and $Y$, we denote by $X + Y$ the element-wise addition, by $X - Y$ the element-wise subtraction, and by $X \le Y$ the element-wise inequality.
We denote by $X^k$ the array produced by shifting $X$ by $k$ rightwards if $k \ge 0$ and the array produced by shifting $X$ by $|k|$ leftwards if $k < 0$; the overflowed entries are discarded.
Let $x, y$ be the two children of $v$.
We make a formula to calculate $\DP[v][P][Q]$ from the $\DP$ values for children.
Let the cost of the edges $(v,x),(v,y)$ in $\mathcal{I}'$ be $d_x,d_y$.
Then, $\DP[v][P(v)][Q(v)]$ is $\texttt{true}$ if and only if
\begin{itemize}
    \item there are arrays $P(x),Q(x),P(y),Q(y),R(x),R(y),S_1,S_2,W_1,W_2$ of integers whose lengths are $K+1$ such that 
    \item $S_1+S_2$ is $(1,0, \dots, 0)$ if there is a user on $v$ and $(0, \dots, 0)$ otherwise, and 
    \item the sum of all elements in $W_1+W_2$ is zero or at least $r$ if there is a facility on $v$ and zero otherwise, and
    \item if $W_1+W_2$ is nonzero, the sum of indices of last nonzero elements of $W_1$ and $W_2$ are at most $K$, and
    \item $R(x)\leq P(x)^{d_x},Q(y)^{d_y}$ and $R(y)^{d_y}\leq P(y)^{d_y},Q(x)^{d_x}$, and
    \item $\DP[x][P(x)][Q(x)]=\DP[y][P(y)][Q(y)]=\texttt{true}$, and 
    \item $p(x)_i=0$ for $i>K-d_x$, $q(x)_i=0$ for $i<d_x$, $p(y)_i=0$ for $i>K-d_y$, $q(y)_i=0$ for $i<d_y$, and
    \item $P(v)=P(x)^{d_x}+P(y)^{d_y}-R(x)-R(y)+S_1-S_2-W_1$, and
    \item $Q(v)=Q(x)^{-d_x}+Q(y)^{-d_y}-R(x)-R(y)+W_2$.
\end{itemize}
The meaning of the auxiliary variables $R(x), R(y), S_1, S_2, W_1, W_2$ are as follows.
\begin{itemize}
    \item The $i$-th entry of $R(x)$ (resp. $R(y)$) denotes the number of users in $T_x$ (resp. $T_y$) who are distant from $v$ by distance $i$ and assigned to the facility in $T_y$ (resp. $T_x$).
    \item $S_1$ and $S_2$ decide whether we assign the user on $v$ to an open facility in $T_v$ or remain unassigned.
    \item The $i$-th entry of $W_{1}$ (resp. $W_2$) denotes the number of users in $T_v$ (resp. outside of $T_v$) who are assigned to the facility on $v$ and distant from $v$ by distance $i$.
\end{itemize}
We can enumerate all the possibilities of the arrays in polynomial time.
Thus, the total time complexity is polynomial.
We can reconstruct the solution by storing which candidates of transitions are chosen, so we achieved to construct an algorithm what we wanted.
This gives a proof of Theorem~\ref{ptas}.


\subsection{Variants}

Our technique can be used for other variants of the $r$-gathering problems. In {\it $(r,\epsilon)$-gathering problem} \cite{aggarwal2010achieving}, we do not need to assign at most $\epsilon$ factor of users. We can construct an algorithm to solve it, just by adding the number of ignored users in $T_v$ to DP states of vertex $v$. Note that, this problem is also NP-hard and does not admit FPTAS, because we can convert $r$-gathering instance to equivalent $(r,\epsilon)$-gathering instance, just by adding the proper number of users on sufficiently far points.

We can treat the constraint on the number of open facilities just by adding the number of open facilities in $T_v$ to DP states of vertex $v$. Note that, this problem is also NP-hard and does not admit FPTAS because in the gadget construction described in our another paper \cite{wareware} we only have to decide whether there is a solution with $2d+1$ clusters, where $d$ is the number of ``long legs'' on a spider.

Here we give a theorem to conclude this subsection.

\begin{theorem}
Both min-max $(r,\epsilon)$-gathering and $r$-gathering with constraints on the number of open facilities admit strongly polynomial time approximation schema.
\end{theorem}

We can also straightforwardly combine these additional states to solve combined problems.

\section{Polynomial-Time Algorithms for other variants}

In contrast to the min-max $r$-gathering, there are variants which can be solved in polynomial-time in tree. In this section, we introduce them.

\subsection{min-sum $r$-Gathering and Lower Bounded Facility Location Problem}

Now we consider the other objective function -- not min-max, but min-sum. We can also introduce the cost to open facility $c(f)$ for each facility $f$: the total cost is the sum of the distance between users and assigned facilities, and the sum of $c(f)$ over all open facilities. In this situation, the problem is so-called {\it lower bounded facility location problem}~\cite{svitkina2010lower}. For the general metric case, $448$-approximation algorithm was given in~\cite{svitkina2010lower}. Later, the approximation ratio is improved to $82.6$~\cite{ahmadian2012improved}.

Unlike the min-max case, we can solve this problem exactly on a tree in polynomial time. 
For each vertex $v$ and an integer $x$, such that $-|\mathcal{U}|\leq x \leq |\mathcal{U}|$, let us define the value $\DP[v][t]$ by the minimum total cost in following situation.

\begin{itemize}
    \item If $t\geq 0$, all but $t$ users in $T_v$ are assigned to facilities in $T_v$, all open facilities in $T_v$ has at least $r$ users, and we will assign remaining $t$ users to facilities out of $T_v$. In other words, $t$ users go upwards from $v$, and no users go downwards to $v$.
    \item Otherwise, all users in $T_v$ are assigned to facilities in $T_v$, and we will assign additional $|t|$ users out of $T_v$ to the facilities in $T_v$. In other words, $|t|$ users go downwards to $v$, and no users go upward from $v$.
\end{itemize}

We want the value $\DP[\root][0]$. Following observation ensures we can get an optimal solution by calculating DP values in a bottom-up way.

\begin{lemma}
There is an optimal solution, that for each edge $e$, all users who pass through the edge $e$ when they go to the assigned facilities pass through $e$ in the same direction.
\end{lemma}

\begin{proof}
Assume the users $u,u'$ go to the facilities $f,f'$, respectively, and they pass through the edge $e$ in the opposite direction. Then, we can decrease the sum of the number of edges the user pass through among all users, by reassigning $u$ to $f'$ and $u'$ to $f$, without increasing the total cost and breaking feasibility.
\end{proof}

Let us write down the transitions. Denote two children of $v$ by $x,y$, and distance between $x,y$ and $v$ by $d_x,d_y$. We also denote the number of users on $v$ by $u_v$. Then, $\DP[v][t]$ is calculated by
\[
\min_{k}(\DP[x][k]+\DP[y][t-u_v-k]+|k|d_x+|t-u_v-k|d_y)
\]
if $v$ contains no facilities. If $v$ contains a facility $f$, we also decide whether to open $f$. Thus, we additionally take a minimum to the value $c(f)+\min_{k\geq t+r}\DP[v][k]$. We can implement this algorithm to work in $O(|V(T)||\mathcal{U}|^2)$ time. Since $|V(T)|=O(|\mathcal{U}|+|\mathcal{F}|)$, we get the following theorem.

\begin{theorem}
min-sum $r$-gathering problem and lower bounded facility location problem on a tree admit an exact $O((|\mathcal{U}|+|\mathcal{F}|)|\mathcal{U}|^2)$ time algorithm.
\end{theorem}

\subsection{Proximity Requirement}

In real applications, it is natural to assume that users go to their nearest open facilities. This requirement is called {\it proximity requirement}. It is discussed in Armon's paper~\cite{armon2011min} for min-max $r$-gathering problem and they gave a $9$-approximation algorithm. We assume that for all user $u$, there is no tie among the distances from $u$ to the facilities. That ensures the users uniquely determine the facility that they go. Especially, there is a positive distance between two distinct facilities.

Unlike the vanilla $r$-gathering, We can solve this problem exactly in polynomial time on a tree. The key observation is the following fact.

\begin{lem}
Assume that the user $u,u'$ go to the facility $f,f'$, respectively, in a feasible solution. If $u-f$ path and $u'-f'$ path have a common point, $f=f'$.
\end{lem}

\begin{proof}
Denote this common point by $c$. Since $d(u,f)\neq d(u,f')$, $d(c,f)\neq d(c,f')$ holds. Without loss of generality, we can assume that $d(c,f)<d(c,f')$. It means both $u$ and $u'$ should go to $f$.
\end{proof}

By the above lemma, we can argue that if there are two users who go to the same facility, so do all the users between them. From now, we construct an algorithm by dynamic programming.

For each vertex $v$, facility $f$, and integer $0\leq t\leq r$, we calculate the value $\DP[v][f][t]$, which represents the minimum possible cost to assign all users in $T_v$ and decide whether to open each facilities in $T_u$ and $f$, in situation
\begin{itemize}
    \item there are at least $t$ users assigned to $f$,
    \item the nearest open facility from $v$ is $f$,
    \item users in $T_v$ is assigned to the facilities in $T_v$ or $f$,
    \item and all open facilities in $T_v$ but $f$ have at least $r$ users.
\end{itemize}
If there is no solution which satisfies above conditions, this value is $\infty$. We calculate these values in a bottom-up way. 

We want the minimum value $\DP[\root][f][r]$ among all facility $f$. Let us write down the transitions. Let two children of the vertex $v$ be $x$ and $y$, and the number of users on vertex $v$ be $u_v$. Let $cost(v,f)$ be $d(v,f)$ when there are users on $f$ and $0$ when there is no user on $f$. $\DP[v][f][t]$ is calculated by the following minimum. 
\begin{itemize}
    \item $\max(\DP[x][f][k],\DP[y][f][l],cost(v,f))$ for all $k+l+u_v\geq t$. That corresponds to the case remaining users in $T_x,T_y$ are assigned in $f$.
    \item $\max(\DP[x][f][t'],\DP[y][f'][r],cost(v,f))$ for all $t'+u_v\geq t$ and facility $f'$, which satisfies $d(v,f)\leq d(v,f')$ and $d(y,f)\geq d(y,f')$. That corresponds to the case remaining users in $T_x$ are assigned to $f$ and we finish to choose users assigned to $f'$.
    \item $\max(\DP[x][f'][r],\DP[y][f][t],cost(v,f))$ for all $t'+u_v\geq t$ and facility $f'$, which satisfies $d(v,f')\geq d(v,f')$ and $d(x,f')\leq d(x,f)$. That corresponds to the opposite case described above.
    \item $\max(\DP[x][f_x][r],\DP[y][f_y][r],cost(v,f))$ for all facilities $f_x,f_y$, which satisfies $d(v,f_x)\geq d(v,f)$ and $d(v,f_y)\geq d(v,f)$. That corresponds to the case which we finish to choose the users assigned to $f_x,f_y$.
\end{itemize}
We can calculate all these transitions in $O((r+|\mathcal{F}|)^2)$ time for each vertex $v$, so we can solve this problem in $O(|V(T)|(r+|\mathcal{F}|)^2)$ time. Note that, min-sum version of this problem can be solved in the same way. Here we conclude this subsection by the following theorem.

\begin{theorem}
min-max and min-sum $r$-gathering with proximity requirement admit an exact $O((|\mathcal{U}|+|\mathcal{F}|)(r+|\mathcal{F}|)^2)$ time algorithm.
\end{theorem}

\bibliographystyle{plainurl}
\bibliography{bib.bib}

\end{document}